\documentclass[a4paper,12pt]{article}

\usepackage[margin=1.5cm]{geometry} 
\usepackage{authblk}
\usepackage{amsmath}
\usepackage{amsthm}\usepackage{enumerate}
\usepackage{amssymb}
\usepackage[linesnumbered,ruled,lined]{algorithm2e}\newtheorem{case}{Case}\newtheorem{observation}{Observation}
\newcommand{\passive}{$passive$}

\newcommand{\jump}{$jump$}
\usepackage{hyperref}

\theoremstyle{plain}
\newtheorem{theorem}{Theorem}

\newtheorem{lemma}[theorem]{Lemma}

\theoremstyle{definition}
\newtheorem{definition}[theorem]{Definition}

\newtheorem{remark}[theorem]{Remark}

\usepackage{graphicx,subfig}
\usepackage{graphicx, color}

\usepackage{mathrsfs} 

\title{Optimal Dispersion of Silent Robots in a Ring}\date{}
\author[1]{\small Bibhuti Das \thanks{dasbibhuti905@gmail.com}}
\author[1]{\small Barun Gorain \thanks{barun@iitbhilai.ac.in}}
\author[2]{\small Kaushik Mondal \thanks{kaushik.mondal@iitrpr.ac.in}}
\author[3]{\small Krishnendu Mukhopadhyaya \thanks{krishnendu@isical.ac.in}}
\author[4]{\small Supantha Pandit \thanks{pantha.pandit@gmail.com}}
\affil[1]{\footnotesize Indian Institute of Technology Bhilai, India}
\affil[2]{\footnotesize Indian Institute of Technology Ropar, India}
\affil[2]{\footnotesize Indian Statistical Institute Kolkata, India }
\affil[2]{\footnotesize Dhirubhai Ambani Institute of Information and Communication Technology, Gandhinagar, India}
\begin{document}
	\maketitle
	
	\begin{abstract}
		Given a set of co-located mobile robots in an unknown anonymous graph, the robots must relocate themselves in distinct graph nodes to solve the {\it dispersion} problem. In this paper, we consider the {\it dispersion} problem for {\it silent} robots \cite{gorain2024collaborative}, i.e., no direct, explicit communication between any two robots placed in the nodes of an oriented $n$ node ring network. The robots operate in synchronous rounds. The {\it dispersion} problem for {\it silent} mobile robots has been studied in arbitrary graphs where the robots start from a single source. In this paper, we focus on the {\it dispersion} problem for {\it silent} mobile robots where robots can start from multiple sources. The robots have unique labels from a range $[0,\;L]$ for some positive integer $L$. Any two co-located robots do not have the information about the label of the other robot. The robots have {\it weak multiplicity detection} capability, which means they can determine if it is alone on a node. The robots are assumed to be able to identify an increase or decrease in the number of robots present on a node in a particular round. However, the robots can not get the exact number of increase or decrease in the number of robots. We have proposed a deterministic distributed algorithm that solves the {\it dispersion} of $k$ robots in an oriented ring in $O(\log L+k)$ synchronous rounds with $O(\log L)$ bits of memory for each robot. A lower bound $\Omega(\log L+k)$ on time for the dispersion of $k$ robots on a ring network is presented to establish the optimality of the proposed algorithm.\vspace{1cm}
		
	\end{abstract}

	\section{Introduction}
In distributed computing, the {\it dispersion} problem for mobile robots in a graph has recently become very popular. The goal of the {\it dispersion} problem is starting from one or multiple source nodes; a set of mobile robots must be placed in the graph nodes by ensuring that no two robots are co-located. The {\it dispersion} problem was first introduced by Augustine and Moses Jr. \cite{AugustineM18} for a set of mobile robots. They investigated the problem by focusing on minimizing the memory required by each robot and the number of rounds required to achieve
dispersion in an arbitrary graph. Over the years, the {\it dispersion} problem has been well-studied in the literature over various models \cite{AgarwallaAMKS18,Das2020,KshemkalyaniMS19,KshemkalyaniMS20icdcn,KshemkalyaniMS20walcom,MollaMM20,MollaM19,ShintakuSKM20}. As the {\it dispersion} problem aims at distributing a group of robots it has got the potential for several practical applications. One such application is related to electric vehicles and charging stations. In order to reduce the waiting time of a vehicle in a charging station, the cars can be spread over the charging stations by ensuring that each charging station is occupied by one car. 

The {\it dispersion} problem is closely related to several other well-studied problems on a graph network such as exploration \cite{BrassCGX11,BrassVX14,DereniowskiDKPU15}, scattering \cite{BarriereFBS11,ElorB11,mondal2022uniform,poudel2020fast,ShibataMOKM16}, load balancing \cite{berenbrink2009new,cybenko1989dynamic,sauerwald2012tight,subramanian1994analysis,xu1992analysis} etc. 
Dispersion of robots is very similar to the robot scattering or
uniform-deployment problem on graphs \cite{BarriereFBS11,ElorB11,mondal2022uniform,poudel2020fast,ShibataMOKM16}, both of which require the spreading of robots uniformly in the graph. The dispersion problem is closely related to the exploration problem where $n$ robots start at a given node and must explore the graph in
as few rounds as possible \cite{BrassCGX11,BrassVX14,DereniowskiDKPU15}. 
For the dispersion problem, if the number of robots is equal to the number of nodes, then every possible solution to the dispersion problem would also solve the $n$ robot exploration under the same set of conditions. 
The dispersion problem can be viewed as a variant of load balancing on graphs \cite{berenbrink2009new,cybenko1989dynamic,sauerwald2012tight,subramanian1994analysis,xu1992analysis}, wherein the nodes usually begin with an arbitrary amount of load and are required to transfer load using edges until each node has nearly the same amount.

\subsection{Motivations}
Our work is motivated by the result of gathering by Bouchard et al. \cite{BouchardDP20}. The authors have solved the problem of gathering robots without communication capability. Gorain et al. \cite{gorain2024collaborative} investigated the dispersion problem on arbitrary graphs considering the same model described in the paper \cite{BouchardDP20}. In particular, they assumed that a robot could identify the following: (1) whether it is alone on a node, (2) whether the number of robots changes at the node compared to the previous round. They proposed a deterministic algorithm that achieves dispersion on any arbitrary graph in time $O (k \log L + k^2 \log \Delta)$ where $\Delta$ is the maximum degree of a node in the graph. Each robot uses $O(\log L + \log \Delta)$ additional memory, starting from a single source. The memory $O(\log L + \log \Delta)$ required in  \cite{gorain2024collaborative}  optimal. However, the authors did not comment anything about the optimality of the time of dispersion in this model. We investigate the dispersion problem with silent robots on the class of cycles and proposed an optimal solution even if the mobile robots were initially placed in multiple source nodes.

\subsection{Our Contributions}

\begin{enumerate}

\item This work considers the dispersion problem starting from multiple sources. There is no restriction on the initial positions of the robots as well as the number of co-located robots on a multiplicity. Our proposed algorithm works without direct, explicit communication among the robots. 
\item Our proposed algorithm solves the dispersion problem within $O(\log L+k)$ synchronous rounds with $O(\log L)$ bits of persistent memory for each robot.
\item We have shown that dispersion of silent mobile robots requires at least $\Omega(\log L+k)$ synchronous rounds on a ring.
\end{enumerate}

\section{Model and Problem Definition}
The robots are located on the nodes of an anonymous $n$ node ring denoted by $\mathscr{R}=(V,\; E)$. The nodes of the graphs are anonymous. However, the edges incident to a node $v$ are labeled 0 and 1. The ring is oriented, i.e., the robots have a global notion of clockwise and counter-clockwise directions. As a consequence, the labeling of the ports has the same orientation. For a node $v$, the adjacent node connected through port 0 is said to be a predecessor of $v$ and denoted by $pre(v)$. Similarly, the adjacent node connected through port 1 is said to be a successor of $v$ and denoted by $succ(v)$. Initially, multiple nodes in $\mathscr{R}$ can contain more than one robot, i.e., a multiplicity. There is no restriction on the initial placement of such multiplicity nodes. 

The robots move in synchronous rounds, and at most, a robot can traverse one edge in every round. Each round is comprised of two different stages. In the first stage, the robot at a node $v$ completes all the local computations. In the second stage, based on the computations done in the previous stage, either move along one of the edges incident to $v$ or stay at $v$. When two robots choose to move in the same round along the same edge from opposite endpoints, neither robot is capable of detecting the other's movement. In other words, a robot can detect the presence of other robots only at nodes in the graph, not at edges.

Every mobile robot has a unique ID in the $[0,\;L]$ range, where $L \ge k$. A mobile robot is aware of its ID and the value of $L$ but is uninformed of the identities of other robots. The robots are silent, i.e., no two robots can directly communicate. However, at any round, a robot is capable of detecting the following local activities:

\begin{enumerate}
\item The robots have {\it weak multiplicity detection} capability. If the robot is present at a node $v$ in a round and no other robot is co-located with it, then it can detect that it is alone in this round. A local variable $alone=true$ represents the robot's detection of this activity. 
 \item If the robot stays at a node $v$ in a round $t$, and the number of robots co-located at $v$ increases at the end of round $t$, then the robot can detect in round $t+1$ that the number of co-located robots increased in the previous round. A local variable $increase=true$ represents the detection of this activity by the robot in round $t+1$. If no such activity is detected, then $increase=false$.
 
\item If the robot stays at a node $v$ in a round $t$, and the number of robots co-located at $v$ decreases at the end of round $t$, then the robot can detect in round $t+1$ that the number of co-located robots increased in the previous round. A local variable $decrease=true$ represents the detection of this activity by the robot in round $t+1$. If no such activity is detected, then $decrease=false$.
\end{enumerate}
Note that the robots can not identify the exact number of co-located robots that increased or decreased. Also, if there is an equal number of arriving and leaving robots at node $v$, then the robot $M$ can not detect anything, and both the variables $increase$ and $decrease$ will be set as false in the next round. \\

 Let $MaxSize$ denote the length of the binary representation of the maximum ID of a robot which is $L$. As $L$ is known to all the robots, $MaxSize=\lfloor\log L\rfloor+1$ is a global knowledge. 
 Without loss of generality, we may assume that the ids of each robot is of size $MaxSize$. Otherwise, every robot can append necessary numbers of zeros in the front of its binary label to make the size of the label equals to $MaxSize$.

\subsection{Related Work} Augustine and Moses Jr. \cite{AugustineM18} were the first to introduce the {\it dispersion} problem. The authors investigated the problem with the assumption that the number of robots (say $k$) is equal to the number of vertices (say $n$). The problem has been investigated in paths, rings, trees, and arbitrary graphs. They proved that, for any graph $G$ of diameter $D$, any deterministic algorithm must take $\Omega(\log k)$ bits of memory by each robot and $\Omega(\log D)$ number of rounds. For arbitrary graphs, they provided an algorithm that requires $O(m)$ rounds where each robot requires $O(n\log n)$ bits of memory. For paths, rings, and trees, they proposed algorithms for each kind of graph such that each robot requires $O(\log n)$ bits of memory and takes $O(n)$ rounds. Furthermore, their proposed algorithm for rooted trees takes $O(D^2)$ rounds, and each robot requires $O(\Delta + \log n)$ bits of memory.

In the paper, \cite{KshemkalyaniF19}, Kshemkalyani and Ali proposed five different {\it dispersion} problems on graphs. Their first three algorithms require $O(k \log \Delta)$ bits for each robot and $O(m)$ steps running time, where $m$ is the number of edges and $\Delta$ is the degree of the graph. These three algorithms differ in whether they use the synchronous or the asynchronous system model and in what, where, and how data structures are maintained. Their fourth algorithm considers the asynchronous model and uses $O(D \log \Delta)$ bits of memory at each robot and $O(\Delta ^D)$ rounds, where $D$ is the graph diameter. Their fifth algorithm under the asynchronous model uses $O(\max(\log k, \log \Delta ))$ bits memory at each robot and $O((m-n)k)$ rounds. In \cite{KshemkalyaniMS19}, Kshemkalyani et al. provided a deterministic algorithm in arbitrary graphs in a synchronous model that requires $O(\min(m, k \Delta) \log k)$ rounds and $O(\log n)$ bits of memory by each robot. 

The {\it dispersion} problem was first studied on a grid graph by Kshemkalyani et al. \cite{KshemkalyaniMS20walcom}. They provided two deterministic algorithms: (i) for the local communication model (where robots can only communicate with other robots that are present at the same node.) and (ii) for the global communication model (where a robot can communicate with any other robot in the graph possibly at different nodes (but the graph structure is not known to robots)). For the local communication model, they proposed an algorithm that requires $O(\min(k, \sqrt{n}))$ rounds, and each robot uses  $O(\log k)$ bits of memory. For the local communication model, their algorithm requires $O(\sqrt{k}))$ rounds, and each robot uses  $O(\log k)$ bits of memory. An extension of this work on general graphs is considered by Kshemkalyani et al. \cite{KshemkalyaniMS20icdcn}. They proposed an algorithm for arbitrary graphs using DFS traversal that requires $O(\min(m, k\Delta))$ rounds and uses $\Theta(\log(\max(k, \Delta)))$ bits memory at each robot. They also proposed a BFS
traversal-based algorithm for arbitrary graphs that requires $O(\max(D, k)\delta(D + \Delta))$ rounds and uses $O(\max(D, \Delta \log k))$ bits memory at each robot. Further, they proposed an algorithm for arbitrary trees using
BFS traversal that requires $O(D \max(D, k))$ rounds and uses $O(\max(D, \Delta \log k))$ bits memory at each robot. Agarwalla et al. \cite{AgarwallaAMKS18} explored the {\it dispersion} problem on dynamic rings. In the paper  \cite{MollaMM20}, Molla et al. introduced fault-tolerant in {\it dispersion} problem in a ring in the presence of Byzantine robots.

Molla et al. \cite{MollaM19} used randomness in the {\it dispersion} problem and proposed an algorithm assuming $O(\log \Delta)$ bits of memory for each of the robots. For any randomized algorithm for the {\it dispersion} problem, a matching lower bound of $\Omega(\log \Delta)$ bits is presented. The authors also investigated a generalized version of dispersion, namely
$k$-dispersion problem where $k > n$ robots must be dispersed over $n$ nodes by ensuring that at most $\dfrac{k}{n}$ robots are placed at each node. Das et al. \cite{Das2020} provided an optimal algorithm such that each robot uses $O(\log \Delta)$ bits of memory.

\section{Dispersion from Multiple Sources}

\subsection{High level Idea}
{\color{black}{Before going to the details of the proposed algorithm, as a warm-up, we present an informal description of dispersing mobile robots from a single source in an oriented ring. Initially, all the robots are located at the source node $s$. Since the robots have distinct IDs, in the binary representation of the IDs of two robots, from right to left, at least one bit must be different. Based on this fact, the robots process the bits of the binary representations of the IDs one by one from right to left, and for each processed bit by a robot with a value of 1, the robot moves clockwise one step. If the bit is 0, the robot does not move. Hence, two robots, having different $j^{th}$ bits, will be separated while processing the $j^{th}$ bit of their IDs. Therefore, the set of robots will be separated into two groups: the group of robots with $j^{th}$ bit 0 stays in the previous node, and the group of robots with the $j^{th}$ bit 1 moves clockwise 1 step. We call such an event a {\it split}. There may be an integer $j$ such that all the robots have the same $j^{th}$. Hence, split only happens if at least two robots have different $j^{th}$ bit. Thus, if the length of the binary representation of the maximum ID is $\ell$ when $\ell$ bits are processed, every robot will be separated from each other. 

In order to implement the above idea, we must ensure that if two groups of robots are separated once, they must not be co-located again. If this happens, there may be two robots, one from each of these two separated groups, for which the bits in the rest of the IDs are the same. Hence, they will not be separated using the above process. 

In order to achieve this, we `activate' alternate groups in a particular phase. Specifically, among two consecutive groups, exactly one group of robots processes their bits in a particular phase, and the next group remains inactive. As a consequence, it would ensure that the split happens only in one group, and the other group of robots can detect this split. If a split happens, the next group, which was inactive, moves one step forward in the next phase to create a space for this new group that has been created due to the split in the previous phase. In this way, no two groups, which have been separated before due to a split, will ever be merged again. 

{\color{black}{If we allow the robots to be active in two adjacent nodes simultaneously, the robots may not identify the split. For example, consider two groups of robots $G_1$ and $G_2$ present on adjacent nodes $v_1$ and $v_2$, respectively. Suppose the robots on both nodes become active at the same round, processing the $j^{th}$ bit. Next, consider the following two different scenarios:
\begin{enumerate}
    \item No split happened at the nodes $v_1$ and $v_2$. For both groups, the $j^{th}$ is 0 for all the robots.
    \item A split happened at the nodes $v_1$ and $v_2$. Consider the case when a split happened for both groups such that the number of incoming robots is exactly equal to the outgoing robots at the node $v_2$. 
\end{enumerate} However, the robots of $G_2$ can not possibly distinguish between the above two cases. In particular, the robots of $G_2$ can not identify the split for the second case. To proceed, the robots must check the $(j+1)^{th}$ bit for a possible split for the first case. In doing so, robots from $G_1$ and $G_2$ co-located on $v_2$ would start processing the $(j+1)^{th}$ without realizing they are from different groups. As discussed earlier, there may be two robots that might not be separated in the future.  }}

The main difficulty in implementing the above idea for the un-rooted case is determining which group will be active or inactive in a particular phase (A phase is a collection of rounds. The details will be defined later). In the rooted case, in the odd phases, groups positioned at odd  distances from the source are active, and in the even phases, even distanced groups are active. In case of arbitrary start, having multiple sources, two consecutive groups (from two different sources) may become active in the same phase, which creates a problem in identifying a split. To avoid such a scenario, we follow the steps described below.

\begin{itemize}
\item We ensure that no two groups are consecutive initially before they start processing the bits of their labels. To do so, first, we merge consecutive groups into a single group. This is done in the following way. Every group tries to elect a leader robot by processing the bits of their labels. After every split, the robots in each group will move to the predecessor node to check whether a robot becomes alone or not due to the split. Since only the first group of robots of a chain of consecutive groups have an empty neighbor in the backward direction, only this group will be able to elect the leader successfully. Once the leader is elected, the leader moves forward to visit all the groups of robots until it finds an empty node. When visited by a robot in a specific round, the other groups understand this as a signal to merge and merge with the last group in the chain.

\item  Once the merging is done, the algorithm executes as explained in the rooted case. However, the leader, who was previously elected, moves forward in each phase to ensure that an empty node is present to expand the chain, which is now created due to splits. Whenever the leader finds another group of agents ahead, the group of agents synchronizes their relative positions (odd or even) with respect to the previous group.  
\end{itemize}
The details of the above process are explained later while describing the algorithm.
}
}

\begin{definition}
Let $G_i$ denote the group of robots on the node $v_i$. A set of group of robots $C_j=\lbrace G_1,G_2,\ldots,G_p\rbrace$ for some $j,p>0$ is said to form a chain of groups if $\forall i\in \lbrace 2,\ldots,p-1\rbrace,\; pred(v_i)=v_{i-1}$ and $succ(v_i)=v_{i+1}$. 
\end{definition}
A chain of groups refers to consecutive groups of robots without any empty nodes. Initially, there can be multiple chains of groups. While executing our proposed algorithm $AlgorithmMultiStart$, all the robots belonging to the same chain of groups are merged on a single node. Later on, the robots would form a chain of groups of robots due to a split of robots by the execution of $AlgorithmMultiStart$. 

During any execution of $DispersionMultiStart$, a robot checks its status and variables to decide its next action. All the robots have the following variables: $move$, $advance$, $leader$, $start$, $settle$, $increase$, $decrease$, and $alone$ stored in its memory. 
\begin{enumerate}
    \item $\boldsymbol{move}:$ The variable $move$ can have one of the following values from the set $\lbrace 0,1,2\rbrace$. The variable $move$ is reset to 0 before the start of each phase. During the execution of $DispersionMultiStart$, the variable can be updated to its permitted values from the set $\lbrace 0,1,2\rbrace$ based on the robot movement. 
    \item $\boldsymbol{proceed}:$ The variable $proceed$ can have one of the following values from the set $\lbrace 0,1,2\rbrace$. The variable $proceed$ is not reset to 0 like the variable $move$ before the start of each phase. During the execution of $DispersionMultiStart$, the variable can be updated to its permitted values from the set $\lbrace 0,1,2\rbrace$ based on the robot movement.  
    \item $\boldsymbol{start}:$ The variable $start$ is a binary variable. Let a robot $M$ be at the node $v$, and $w$ be the predecessor of $v$. If all the robots up to $w$ are $idle$ and $M$ is a possible candidate that may become $idle$ at the beginning of its next phase, then $M$ would update the variable $start$ to 1. Note that the $start$ variable is not reset to 0 at the beginning of a phase. 
    \item $\boldsymbol{settle}:$ The $settle$ variable is also a binary variable. Let a robot $M$ be at the node $v$, and $w$ be the predecessor of $v$. If all the robots up to $w$ are $idle$, and $M$ is alone at the node $v$, then the $settle$ is updated to 1. 
    \item $\boldsymbol{alone}:$ If a robot $M$ is the only robot that lies on a node, it updates its $alone$ variable to $true$. Otherwise, the $alone$ variable is set to $false$. 
    \item $\boldsymbol{increase}:$ At the beginning of a round, $increase$ is set to $false$ for all the robots. If some robots arrive at a node $u$ in the round $t$, then for all the robots that were at node $u$ before the arrival of robots in round $t$, the variable $increase$ is set to $true$ at the beginning of round $t+1$. Otherwise, the variable $increase$ remained set to $false$. 
     \item $\boldsymbol{decrease}:$ At the beginning of a round, $decrease$ is set to $false$. If some robots leave at a node $u$ in the round $t$, then for all the robots that would remain at node $u$ before the arrival of robots in round $t$, the variable $decrease$ is set to $true$ at the beginning of round $t+1$. Otherwise, the variable $decrease$ is remained set to $false$. 
     \item $\boldsymbol{leader}:$ Initially, the variable $leader$ is set to $false$. If a robot $M$ identifies as a leader, the variable $leader$ is set to $true$. 
      \item $\boldsymbol{advance}:$ Initially, the variable $advance$ is set to 0 for all the robots. If the variable $leader=true$, then at $MaxSize+1$ rounds, a leader $M$ updates the variable $advance$ to 1.

\end{enumerate}

{\color{black}{  Each robot have a variable called {\it status} which can be either $leaderelection, \;activemerge,\\ \;activedisperse, \;passive, \;wait, \;jump,\; or\; idle$. Initially, all the robots have the status $leaderelection$. A robot $M$ with status $X \in \{activemerge, \;activedisperse,\; passive,\; wait, \\jump,\; idle\}$ calls the procedure $X(M)$ at the start of any phase. During the execution of a subroutine $X(M)$, if a robot with status $X$ updates its status to $Y$, then it continues the execution of $X(M)$ until the current phase completes and in the next phase it starts executing $Y(M)$. }}The proposed algorithm runs in phases where each phase consists of 19 synchronous rounds. Based on the status an robot $M$ executes either $leaderElection(M,i)$ $ActiveMerge(M)$, $ActiveDisperse(M)$, $Passive(M)$, $Wait(M)$, and $Jump(M)$. A robot with status $idle$ never participates in any algorithm.

\begin{center}
\begin{tabular}{|c|c|c|c|c|c|c|c| } 
\hline
status &leaderelection & activemerge &  activedisperse & passive & wait & jump& \color{blue}{Leader}\\
\hline
 Round 1 & Yes  & No &No&No&No&No&\color{blue}{No}\\
 \hline
Round 2 &Yes  & No &No&No&No&No&\color{blue}{No}\\
\hline
Round 3 & Yes  & No &No&No&No&No&\color{blue}{No}\\
\hline
Round 4 & Yes  & No &No&No&No&No&\color{blue}{No}\\
\hline
Round 5 & Yes  & No &No&No&No&No&\color{blue}{Yes}\\
\hline
Round 6 & No  & Yes &No&No&No&No&\color{blue}{Yes}\\
\hline
Round 7 & No  & Yes &No&No&No&No&\color{blue}{Yes}\\
\hline
Round 8 & No  & Yes &No&No&No&No&\color{blue}{No}\\
\hline
Round 9 & No  & No & Yes & Yes &No&No&\color{blue}{Yes}\\
\hline
Round 10 & No  & No & Yes & Yes &No&No&\color{blue}{Yes}\\
\hline
Round 11 & No  & No & Yes & Yes &No&No&\color{blue}{Yes}\\
\hline
Round 12 & No  & No & Yes & Yes &No&No&\color{blue}{No}\\
\hline
Round 13 & No  & No & Yes &No &No&No&\color{blue}{No}\\
\hline
Round 14 & No  & No & Yes &No & Yes &No &\color{blue}{No}\\
\hline
Round 15 & No  & No & Yes & Yes &No &No &\color{blue}{No}\\
\hline
Round 16 & No  & No & No & Yes &No &No &\color{blue}{No}\\
\hline
Round 17 & No  & No & Yes & Yes & Yes & Yes &\color{blue}{No}\\
\hline
Round 18 & No  & No & Yes &No &No &No &\color{blue}{No}\\
\hline
Round 19 & No  & No & Yes &Yes &No &No &\color{blue}{No}\\
\hline
\end{tabular}
\end{center}

\begin{algorithm}[h]
\scriptsize
\KwIn{$M,\;i$}

\tcp{ Round 1:}
\uIf{alone=true and $proceed=0$}
    {
    update leader=true \;
    }
\ElseIf{the $i^{th}$ bit of $l(M)$ is 1} 
    {
     update $proceed=1$ and move through \bf{Port 1} \;\tcp{If at least two robots have different $i^{th}$ bit, a split happens}
    } 
 \tcp{Round 2:}
\If{proceed=0 and decrease=true}
{ update $proceed=2$ and move through \bf{Port 1} \;\tcp{If a split happened, then the robots which stayed back move forward to inform about the split}}
  \tcp{Round 3:}
\uIf{proceed=1 and increase=true or proceed=2}
{move through {\bf Port 0}\;}
\uIf{proceed=1 and increase=false}
{update $proceed=0$ and move through {\bf Port 0}\;}
\tcp{Round 4:}
\uIf{proceed=1}
{move through {\bf Port 0}\;}
\ElseIf{proceed=2}
{stay at the current node\;}
\tcp{Round 5:}
\uIf{alone=true and proceed=1}
{update leader=true, $proceed=0$ and move through {\bf Port 1}\;
}
\ElseIf{alone=false and proceed=1}
{move through {\bf Port 1} and update $proceed=0$\;}
\uIf{i=MaxSize}{update $status=activemerge$\;}
\Else{update i=i+1\;}
\tcp{Round 6 - Round 19}
stay at the current node\;

\caption{$LeaderElection(M,i)$}
\label{algo:LE}
\end{algorithm}

\begin{algorithm}[]
\scriptsize
\KwIn{$M$}
\tcp{Round 1 - Round 5}
stay at the current node\;
\tcp{Round 6:}

\If{leader=true}
{move through {\bf Port 1}\;}
\tcp{Round 7:}
\If{alone=true and leader=true}
{move through {\bf Port 0} and update status=ActiveDisperse\;\tcp{leader identify that merging is complete and move to its previous node}}
\tcp{Round 8:}
\uIf{increase=false, leader=false }
{move through {\bf Port 1}\;}
\ElseIf{increase=true, leader=false}
{update status=ActiveDisperse\;\tcp{robots identify that merging is complete}}
\tcp{Round 9 - Round 19}
stay at the current node\;
\caption{$ActiveMerge(M)$}
\label{algo:AM}
\end{algorithm}

\subsection{Description of $\boldsymbol{AlgorithmMultiStart}$}
In this section, we present a round-wise description of our proposed algorithm $AlgorithmMultiStart$ in a particular phase. 
\begin{enumerate}
    \item {\bf Rounds designated for Leader Election }: Rounds 1-5 are designated for the leader election step. Initially, all the robots have $status=leaderelection$ and $i=MaxSize$.
    \begin{enumerate}
        \item {\bf Round 1:} In this round, only the robots with $\boldsymbol{status=leaderelection}$ participates. The robots process the $i^{th}$ bit of $l(M)$. If the $i^{th}$ bit of $l(M)$ is 1, the robot $M$ updates the variable $proceed=1$ (By checking the variable $proceed$, a robot with $status=leaderelection$ can identify that it has moved one step forward in round 1), and moves through port 1. Otherwise, the robot $M$ stays at the current node. A split will occur if at least two robots have different $i^{th}$ bit. If $alone=true$ and $proceed=0$, then the robot updates $leader=true$.
        \item {\bf Round 2:} In this round, only the robots with $status=leaderelection$ participates. If a split happened in round 1, then the robots that did not move (By checking whether the variable $proceed=0$) can identify $decrease=true$ in this round. As a consequence, they become aware of the split in round 1. However, the robots that moved forward in round 1 ($proceed=1$) can not possibly identify the split. To inform them about the split, the robots with $status=leaderelection$, $proceed=0$, and $decrease=true$ update $proceed=2$ and move through port 1 in this current round. 
        \item {\bf Round 3:} In this round, only the robots with $\boldsymbol{status=leaderelection}$ participates. Due to the movement of robots with $status=leaderelection$, $proceed=0$, and $decrease=true$ in round 2, the robots with $proceed=1$ (moved in round 1) identify $increase=true$. As a consequence, they become aware of the split in round 1. The robots with $proceed=1$ and $increase=true$, together with the robots with $proceed=2$, move through port 0 (they return to the node where the split happened). If a robot with $proceed=1$ identifies $increase=false$, then it becomes aware that a split did not happen in round 1 as the $i^{th}$ bit was 1 for all the robots. In such a scenario ($status=leaderelection$, $proceed=0$ and $decrease=false$), the robots move through port 0 and update $proceed=0$. 
        \item {\bf Round 4:} In this round, only the robots with $\boldsymbol{status=leaderelection}$ participates. If $proceed=1$, then the robots move through port 0. Note that according to the definition of a chain of groups, there exists at least one node (say $v_1$) whose predecessor node is empty. As all the robots have distinct IDs, for the robots in $v_1$, there exists one phase in which only a single robot with $proceed=1$ would move to the empty node in round 4. The robots with $proceed=2$ would remain at the current node till $MaxSize$ rounds (They can be viewed as a disqualified candidate for being a leader). 
        \item {\bf Round 5:} In this round, only the robots with $\boldsymbol{status=leaderelection}$ participates. If a robot with $proceed=1$ identifies $alone=true$, it updates $leader=true$ and moves through port 1. In case the robots identify $alone=false$, they move through port 0 and update $proceed=0$. If $i=MaxSize$, the robots terminate the leader election by updating $status=activemerge$. Else, they update $i=i+1$.
    \end{enumerate}
    \item {\bf Rounds designated for Merging of Chains before Dispersion :} Rounds 6-8 are designated for merging the groups in a chain. 
    \begin{enumerate}
        \item {\bf Round 6:} In this round, only the robots with $\boldsymbol{status=activemerge}$ participates. If a robot identifies $leader=true$ and $status=activemerge$, it moves through port 1 to start the merging. It continues to move forward in this current round until it finds an empty node.  
        \item {\bf Round 7:} In this round, only the robots with $\boldsymbol{status=activemerge}$ participates. In this round, the leader checks whether it has moved to an empty node. If $alone=true$ and $leader=true$, the leader identifies that the merging is complete. It moves through port 0 and terminates the merging by updating $status=activedisperse$.
        \item {\bf Round 8:} In this round, only the robots with $\boldsymbol{status=activemerge}$ participates. In this round, a non-leader robot discovers whether the leader has returned or not in round 7. If a robot identifies $leader=false$ and $increase=false$, it becomes aware that the merging is incomplete. It moves through port 1 to continue merging. In case $leader=false$ and $increase=true$, a robot identifies that the merging is complete and the leader has returned. It updates $status=activedisperse$ and terminates merging. 
    \end{enumerate}
     
     \item {\bf Rounds designated for Dispersion:} Rounds 9-19 are designated for dispersion. 
     \begin{enumerate}
         \item {\bf Round 9:} In this round, only a leader with status $\boldsymbol{activedisperse}$ or $\boldsymbol{passive}$ participates. If the leader identifies $advance=0$ and $alone=false$, it moves through port 1 (to check whether the successive node is empty) and updates $advance=1$. 
         \item {\bf Round 10:} In this round, only a leader with status $\boldsymbol{activedisperse}$ or $\boldsymbol{passive}$ participates. If the leader identifies $advance=1$ and $alone=true$, it moves through port 1 (to check whether the successive node is empty). 
          \item {\bf Round 11:} In this round, only a leader with status $\boldsymbol{activedisperse}$ or $\boldsymbol{passive}$ participates. If the leader identifies $advance=1$ and $alone=true$, it identifies that it has moved to an empty node. It moves through port 0 and updates $advance=0$. The leader becomes aware of the fact that the successive node is empty. 
           \item {\bf Round 12:} In this round, only a non-leader with status $\boldsymbol{activedisperse}$ or $\boldsymbol{passive}$ participates. If a non-leader robot  identifies $increase=true$ due to the movement of a leader in round 11, it moves through port 0 and updates $status=passive$. 
           \item {\bf Round 13:} In this round, only a robot with status $\boldsymbol{activedisperse}$ participates. If it identifies $alone=true$ and $start=0$, then it updates $start=1$. In case, $alone=true$ and $start=0$, then a robot updates $settle=1$ (To remember that it might become $idle$). If the last bit of $l(M)$ is 1, then $M$ updates $move=1$ (to identify that it has moved in round 13) and moves through port 1. If at least two robots have different last bit, a split would occur. 
           \item {\bf Round 14:} In this round, only a non-leader with status $\boldsymbol{activedisperse}$ or $\boldsymbol{jump}$ participates. If a robot with status $activedisperse$ identifies $decrease=true$, it becomes aware of the split in round 13. To inform about the split, the robot updates $move=2$ (to remember that it moved in round 14 to inform about the split) and moves through port 1. If the robot identifies $status=jump$ it moves through port 1 to create space for the robots that had arrived in the previous phase. 
           \item {\bf Round 15:} In this round, only a non-leader with status $\boldsymbol{activedisperse}$ or $\boldsymbol{passive}$ participates. In case a robot with $status=activedisperse$ identifies $move=0$, it realizes that no split happened in round 13 as all the robots had the last bit 0. It updates $status=passive$. Otherwise, if $move=1$ and $increase=false$, the robot becomes aware that no split happened in round 13 as all the robots had the last bit 1. In case $move=2$, it identifies that a split had occurred and it moved to inform about the split. It moves through port 0 and updates $status=passive$. If a robot with $status=passive$ identify $increase=true$ due to a split in round 13, it sets $move=1$ and becomes aware that it must vacate the current node. 
           
           \item {\bf Round 16:} In this round, only a non-leader with status $\boldsymbol{passive}$ participates. If a robot with $status=passive$ identifies $move=1$, it realises that it must vacate the current node by informing the incoming group that they must wait. To inform the incoming group the robot moves through port 0.  
           \item {\bf Round 17:} In this round, only a non-leader with status $\boldsymbol{activedisperse}$ or $\boldsymbol{passive}$ or $\boldsymbol{jump}$ or $\boldsymbol{wait}$ participates. If a robot with $status=activedisperse$ identifies $move=1$ and $decrease=true$, it realizes that it has moved to an occupied node. It updates $status=wait$. In case a robot with $status=activedisperse$ identifies $move=1$ and $decrease=false$, it realizes that it has moved to an empty node. It updates $status=activedisperse$. If a robot with $status=passive$ identifies $move=0$, then it $updates=activedisperse$. Consider the case when a robot with $status=passive$ identifies $move=1$. It learns that a split happened and it has moved in round 16 to inform about the split. It updates $status=jump$  and moves through port 1. If a robot with $status=jump$ identifies $decrease=false$, it learns that there were no robots in the current node when it jumped from the predecessor node. It updates $status=activedisperse$. Consider the case when a robot with $status=jump$ identifies $decrease=true$. it becomes aware of the fact that it jumped to an occupied node. It updates $status=wait$. if a robot has $status=wait$, it updates $status=passive$.
          \item {\bf Round 18:} In this round, only a robot with status $\boldsymbol{activedisperse}$ participates. If a robot with $status=activedisperse$ identifies $settle=1$, then it moves through port 1 (to inform its successive robots it will become $idle$). 
           \item {\bf Round 19:} In this round, only a robot with status $\boldsymbol{activedisperse}$ and $\boldsymbol{passive}$ participates. If a robot with $status=activedisperse$ identifies $settle=1$, then it moves through port 0 and updates $status=idle$. If a robot with $status=passive$ identifies $increase=true$, then it learns that the robot in the predecessor node will become $idle$ and the robot has arrived to inform the same. It updates $set=1$ (To remember that the predecessor robot has become $idle$). 
           
     \end{enumerate}
\end{enumerate}

\begin{algorithm}[]
\scriptsize
\tcp{ Round 1 - Round 8}stay at the current node\;
\tcp{Round 9:}
\uIf{leader=true, advance=0, and alone=false}
{set advance=1 and move through {\bf Port 1}\;}
\tcp{Round 10:}
\If{leader=true and advance=1 and alone=true}
{move through {\bf Port 1}\;}
\tcp{Round 11:}
\If{leader=true and advance=1 and alone=true}
{move through {\bf Port 0} and update $advance=0$\;}
\tcp{Round 12:}
\If{increase=true}
{move through {\bf Port 0} and 
set status=passive\;}
 \tcp{ \bf Round 13:} 
    \uIf{$alone=true$ and $start=0$}
        {update $start=1$\;} 
      \uElseIf{$alone=true$ and $start=1$}  
        {update $settle=1$\;}
    \ElseIf{the last bit of $l(M)$ is 1}
        {
            set $move=1$ and move through {\bf port 1} \tcp{If at least two robots have different last bits, a split happens}
        }
        
     \tcp{ \bf Round 14:}
        \If{$move=0$ and $decrease=true$}
        {
         Update $move=2$ and move through port 1 \tcc{If a split happened in the last round, the nodes that remained will move to inform about the split}  
        }
      
\tcp{ \bf Round 15:}

\uIf{move=0}{
        {$status=passive$ }\tcc{If all robots have the same last bit, then no split happened in round 1, and they all become passive for the next Phase}}
\ElseIf{($move=1$ and $increase=false$) or $move=2$ } 
        {
            
           $status=passive$ and move through port 0 \tcc{either all moved or a subset that moved  in round 1 comes back and becomes passive for the next Phase} 
        }
     \tcp{\bf Round 16}{stay at the current node\;}
      \tcp{\bf Round 17}
{\uIf{$move=1$ and $decrease=true$}
        {
         
        {$status=wait$}\tcc{The active robots learned that they moved to a node occupied by passive nodes. They become wait for the next Phase}
        
        \If{$start=1$}{
        {$start=0$}
      }
      
    }
    \ElseIf{$move=1$ and $decrease=false$ }{
            {set $status=activedisperse$\;}\tcc{The active robots learned that they had moved to an empty node. They remain active for the next Phase}
            
            \If{$start=1$}
                {$start=0$}}}
\tcp{ \bf Round 18:}
\If{$settle=1$}{
    {move through port 1}
}
\tcp{ \bf Round 19:}

\If{$settle=1$}{
    {move through port 0}
    
    {$status=idle$}
}
\caption{$ActiveDisperse(M)$}
\label{activedisperse}
\end{algorithm}

\begin{algorithm}[]
\scriptsize
\tcp{ Round 1 - Round 8}stay at the current node\;

\tcp{Round 9:}
\uIf{leader=true, advance=0, and alone=false}
{set advance=1 and move through {\bf Port 1}\;}
\tcp{Round 10:}
\If{leader=true and advance=1 and alone=true}
{move through {\bf Port 1}\;}
\tcp{Round 11:}
\If{leader=true and advance=1 and alone=true}
{move through {\bf Port 0} and update $advance=0$\;}
\tcp{Round 12:}
\If{increase=true}
{move through {\bf Port 0} and 
set status=passive\;}
\tcp{Round 13-14:}stay at the current node\;
\tcp{ Round 15:} 
\If{increase=true}{set move=1\;\tcc{The robots identify that they need to vacate the current node\;}}
\tcp{ Round 16:} 
\If{move=1}{move through {\bf port 0}\;\tcc{The current node is already occupied. Move through port 0 to inform the incoming group of the same}}
\tcp{ Round 17:}
\uIf{$move=0$}
    {set $status=activedisperse$\;}
    \Else{set $status=jump$ and move through port 1\;\tcp{The robot has moved to inform that after the split the robots have moved to an occupied node}}

           
\tcp{ Round 18:}stay at the current node\;
\tcp{ Round 19:}\If{increase=true}{set start=1\;}
\caption{$Passive(M)$}
\label{passive}
\end{algorithm}

\begin{algorithm}[]
\scriptsize

\tcp{ Round 1 - Round 13}stay at the current node\;
\tcp{Round 14:}
      {move through {\bf port 1}}\tcp{The robot advances one step to give space to the robots that arrived at the current node due to a split in the previous phase}
\tcp{Round 15-16:}
        {stay at the current node\;}   
\tcp{Round 17:}
        \uIf{$decrease=false$}
        {
         Set $status=activedisperse$ \tcp{The node learned that no robot was in the current node when it jumped from the previous node. It changes its status to Active for the next phase}
        }
        \Else{
              {set $status=wait$}\tcp{The node learned that there were some robots in the current node when it jumped from the previous node. It changes its status to wait for the next phase}
              }
\tcp{Round 18 - Round 19}
  {stay at the current node}

\caption{$Jump(M)$}
\label{jump}
\end{algorithm}

\begin{algorithm}[]\scriptsize
\tcp{ Round 1 - Round 16}stay at the current node\;
\tcp{Round 17:}set status=passive \;
\tcp{Round 18 - Round 19}stay at the current node\;
 \caption{$Wait(M)$}
\label{wait}
\end{algorithm}

\subsection{Correctness of $\boldsymbol{DispersionMultiStart}$}
The following two lemmas prove that exactly one leader will be elected for every chain after executing the procedure $LeaderElection$ (Algorithm \ref{algo:LE}).
 \begin{lemma}\label{leader_election_a}
   For some $1\leq j\leq MaxSize$, consider the execution of $LeaderElection(M,j)$. Let $A_j$ be the set of robots at a node $v$ for which $proceed=0$ and $B_j$ be the set of robots for which $proceed=2$ at the beginning of Phase $j$. Also, let $A_j(0)$ (similarly $B_j(0)$ ) and $A_j(1)$ (similarly $B_j(1)$) be the set of robots with $j$-th bit 0 and 1, respectively, in $A_j$ (similarly in $B_j$).
If $|A_j|>1$, then 
\begin{enumerate}
\item   $A_{j+1}=A_j$ and $B_{j+1}=B_j$, when $A_j(1)=A_j$ or $A_j(0)=A_j$.
\item $A_{j+1}=A_j(1)$ and $B_{j+1}=B_j\cup A_j(0)$, otherwise.
\end{enumerate}
 \end{lemma}
 \begin{proof}
We use induction to prove the above statements. Let us assume that $V$ represents the initial set of robots at $v$. Assume that $|V|>1$. To prove that the statement is true for $j=1$, consider the execution of $LeaderElection(M,1)$ at the node $v$. We have the following cases:
\begin{case}$\forall M\in V$, the first bit of $l(M)$ is 0: \normalfont We have $A_1(0)=A_1=V$ and $A_1(1)=\emptyset$. All the robots would stay at node $v$ in round 1 (steps 3-4 of Algorithm~\ref{algo:LE}). The robots would identify $decrease=false$ and stay at node $v$ in round 2 (steps 6-7 of Algorithm~\ref{algo:LE}). In round 3 (steps 9-12 of Algorithm~\ref{algo:LE}), the robots would identify $increase=false$, but as $proceed=0$, they would stay at the current node. The robots would stay at the current node as $proceed=0$ in rounds 4-5 (steps 13-21 of Algorithm~\ref{algo:LE}). The robots would stay at the current node for the rest of the rounds. Hence, We have $A_2=A_1(0)=A_1$ and $B_2=\emptyset$. Thus, the statement is true for $j=1$. 
\end{case}

\begin{case}$\forall M\in V$, the first bit of $l(M)$ is 1: \normalfont We have $A_1(0)=\emptyset$ and $A_1(1)=A_1=V$. All the robots would move to $succ(v)$ in round 1 (steps 3-4 of Algorithm~\ref{algo:LE}). The robots would stay at $succ(v)$ in round 2 (steps 6-7 of Algorithm~\ref{algo:LE}). In round 3 (steps 11-12 of Algorithm~\ref{algo:LE}), the robots would identify $proceed=1$ and $increase=false$. The robots would update $proceed=0$ and move to $v$. For the rest of the rounds, the robots will stay at $v$. We have $A_2=A_1(0)=V$ and $B_2=\emptyset$. Thus, the statement is true for $j=1$. 
\end{case}

\begin{case}$\exists M_1,M_2\in V$, the first bit of $l(M_1)$ is 0 and the first bit of $l(M_2)$ is 1 :\normalfont  We have  $A_1(0)\neq\emptyset$ and $A_1(1)\neq \emptyset$. In the round 1 (steps 3-4 of Algorithm~\ref{algo:LE}), the robots of $A_1(1)$ would move to $succ(v)$ and update $proceed=1$. In the round 2 (steps 6-7 of Algorithm~\ref{algo:LE}), the robots of $A_1(0)$ would move to $succ(v)$ and update $proceed=2$. In the round 3 (steps 9-10 of Algorithm~\ref{algo:LE}), the robots of $A_1(0)\cup A_1(1)$ would move back to the node $v$. In round 4 (steps 13-14 of Algorithm~\ref{algo:LE}), the robots of $A_1(0)$ would move to $pred(v)$, and the robots of $A_1(1)$ would stay at the node $v$. In the round 5 (steps 18-21 of Algorithm~\ref{algo:LE}), the robots of $A_1(0)$ would return to $v$ and update $proceed=0$. The robots will stay at the current node for the rest of the rounds. We have $A_2=A_1(0)$ and $B_2=A_1(1)=B_1\cup A_1(1)$ as $A_1=V$ and $B_1=\emptyset$. Thus, the statement is true for $j=1$. 
\end{case}

Therefore, the statements of the lemma are true for $j=1$. Suppose
that the statement is true for some integer $m > 1$. Consider
the execution of $LeaderElection(M,m+1)$. By
induction hypothesis, just before checking the $(m+1)^{th}$ bit for all the robots of $A_m$, $proceed=0$. We have the following cases:
\setcounter{case}{0}
\begin{case}$\forall M\in V$, the $(m+1)^{th}$ bit of $l(M)$ is 0: \normalfont
We have $A_m(0)=A_m$ and $A_m(1)=\emptyset$. All the robots would stay at node $v$ in round 1 (steps 3-4 of Algorithm~\ref{algo:LE}). The robots would identify $decrease=false$ and stay at node $v$ in round 2 (steps 6-7 of Algorithm~\ref{algo:LE}). In round 3 (steps 9-12 of Algorithm~\ref{algo:LE}), the robots would identify $increase=false$, but as $proceed=0$, they would stay at the current node. The robots would stay at the current node as $proceed=0$ in rounds 4-5 (steps 13-21 of Algorithm~\ref{algo:LE}). The robots would stay at the current node for the rest of the rounds. We have $A_{(m+1)}=A_m(0)=A_m$ and $B_{(m+1)}=B_m=\emptyset$. 
\end{case}

\begin{case}$\forall M\in V$, the $(m+1)^{th}$ bit of $l(M)$ is 1: \normalfont
We have $A_m(1)=A_m$ and $A_m(0)=\emptyset$. All the robots would move to $succ(v)$ in round 1 (steps 3-4 of Algorithm~\ref{algo:LE}). The robots would stay at $succ(v)$ in round 2 (steps 6-7 of Algorithm~\ref{algo:LE}). In round 3 (steps 11-12 of Algorithm~\ref{algo:LE}), the robots would identify $proceed=1$ and $increase=false$. The robots would update $proceed=0$ and move to $v$. For the rest of the rounds, the robots will stay at $v$. We have $A_{m+1}=A_m(0)$ and $B_{m+1}=\emptyset$. Thus, the statement is true for $m+1$. 
\end{case}

\begin{case}$\exists M_1, M_2\in V$, the $(m+1)^{th}$ bit of $l(M_1)$ is 0 and the $(m+1)^{th}$ bit of $l(M_2)$ is 1: \normalfont
We have $A_m=A_m(0)$ and $B_m=A_m(1)$. In the round 1 (steps 3-4 of Algorithm~\ref{algo:LE}), the robots of $A_m(1)$ would move to $succ(v)$ and update $proceed=1$. In the round 2 (steps 6-7 of Algorithm~\ref{algo:LE}), the robots of $A_m(0)$ would move to $succ(v)$ and update $proceed=2$. In the round 3 (steps 9-10 of Algorithm~\ref{algo:LE}), the robots of $A_m(0)\cup A_m(1)$ would move back to the node $v$. In round 4 (steps 13-14 of Algorithm~\ref{algo:LE}), the robots of $A_m(0)$ would move to $pred(v)$, and the robots of $A_m(1)$ would stay at the node $v$. In the round 5 (steps 18-21 of Algorithm~\ref{algo:LE}), the robots of $A_m(0)$ would return to $v$ and update $proceed=0$. The robots will stay at the current node for the rest of the rounds. We have $A_{m+1}=A_m(1)$ and $B_{m+1}=B_m\cup A_m(0)$. 
\end{case}
Therefore, by induction, the lemma statement is true for all $1\leq j\leq MaxSize$.
 \end{proof}

\begin{lemma}\label{leader_election_b}
    Let $C_j=\lbrace G_1,G_2,\ldots,G_p\rbrace$ for some $j,p>0$ be a chain of groups in $\mathscr{R}$ where $G_i$ is a group of robots at node $v_i$. If there are $m>1$ number of robots at the node $v_1$, then after $MaxSize$ phases, there would be a unique leader for $C_j$.  
\end{lemma}
 \begin{proof}
     {\it Proof by contradiction.} Consider the execution of the procedure $LeaderElection$ at node $v_1$. 
     
     We claim that $A_{MaxSize}=1$. To prove this, suppose that  $|A_{MaxSize}|>1$. This implies that there exist two robots $M_1$ and $M_2$  such that after $MaxSize$ rounds, $M_1,M_2\in A_{MaxSize}$. By Lemma  \ref{leader_election_a}, they have the exact same bits for all $j\in\lbrace1,2,\ldots,MaxSize\rbrace$.  This is a
contradiction to the unique labels of the robots.

Next we show that only a robot $ M$ from the group $G_1$ will find $alone=true$ in round 5 (steps 18-19) of the execution of $LeaderElection(M,j)$ for some $j$ and no robot from any other group will find $alone=true$ in the execution of $LeaderElection(M,j)$ for any $j$. To prove this for the robots in $G_1$, let  $|A_{ j}|=1$ and $M \in A_{ j}$. During the execution of  $LeaderElection(M,j)$, this robot in round 1 (steps 3-4 of Algorithm~\ref{algo:LE}), moves to $succ(v_1)$ and updates $proceed=1$. In the round 2 (6-7 of Algorithm~\ref{algo:LE}), the robots of $A_j(0)$ would move to $succ(v_1)$ and update $proceed=2$. In the round 3 (9-10 of Algorithm~\ref{algo:LE}), the robots of $A_j(0)\cup A_j(1)$  move back to the node $v$. In the round 4 (steps 13-16 of Algorithm~\ref{algo:LE}), the robot $M$  move to $pred(v_1)$ and the robots in $A_j(0)$  stay at the node $v_1$. In round 5 (steps 18-19 of Algorithm~\ref{algo:LE}), the robot $M$ identifies $alone=true$ and updates $leader=true$.

Since a robot of a group $G_i$ can set $leader =true$ in the $5^{th}$ round (steps 18-19 of Algorithm~\ref{algo:LE}) by identifying $alone=true$ at the node $pred(v_i)$ in the execution of $LeaderElection$, it is enough to show that for any robot $M\in G_t$, for $t>1$, the node $v_{t-1}$ contains at least one robot of the group $G_{t-1}$ in round 4 (steps 13-16 of Algorithm~\ref{algo:LE}) for every execution of $LeaderElection(M,j)$, $j=1, \ldots, MaxSize$. If this true, then any robot of $G_t$ which is present at $v_{t-1}$ identifies $alone=false$ and hence never sets $leader=true$. 

For some $M\in G_{t-1}$, consider the execution of $LeaderElection(M,k)$ for some $k>1$. If $\exists M$ at $v_{t-1}$ for which $proceed=2$, then this robot will remain at $v_{t-1}$ for all the successive phases till $MaxSize$ phases. As a consequence, any robot of $G_t$ present at $v_{t-1}$ will not be able to identify $alone=true$ in round 5 (steps 18-19 of Algorithm~\ref{algo:LE}). Otherwise, we have the following three possible cases:

\setcounter{case}{0}
\begin{case}$\forall M\in G_{t-1}$, the $k^{th}$ bit of $l(M)$ is 0: \normalfont
We have $A_k(0)=A_k$ and $A_k(1)=\emptyset$. All the robots would stay at node $v_{t-1}$ in round 1 (steps 3-4 of Algorithm~\ref{algo:LE}). The robots would identify $decrease=false$ and stay at node $v_{t-1}$ in round 2 (steps 6-7 of Algorithm~\ref{algo:LE}). In round 3 (steps 9-12 of Algorithm~\ref{algo:LE}), the robots would identify $increase=false$, but as $proceed=0$, they would stay at the current node. The robots would stay at the current node as $proceed=0$ in rounds 4-5 (steps 18-21 of Algorithm~\ref{algo:LE}). The robots would stay at the current node for the rest of the rounds. 
\end{case}

\begin{case}$\forall M\in  G_{t-1}$, the $k^{th}$ bit of $l(M)$ is 1: \normalfont
We have $A_t(1)=A_t$ and $A_t(0)=\emptyset$. All the robots would move to $succ(v)$ in round 1 (steps 3-4 of Algorithm~\ref{algo:LE}). The robots would stay at $succ(v_{t-1})$ in round 2 (steps 6-7 of Algorithm~\ref{algo:LE}). In round 3 (steps 9-12 of Algorithm~\ref{algo:LE}), the robots would identify $proceed=1$ and $increase=false$. The robots would update $proceed=0$ and move to $v_{t-1}$. For the rest of the rounds, the robots will stay at $v_{t-1}$. 
\end{case}

\begin{case}$\exists M_1, M_2\in  G_{t-1}$, the $k^{th}$ bit of $l(M_1)$ is 0 and the $k^{th}$ bit of $l(M_2)$ is 1: \normalfont
We have $A_k=A_k(0)$ and $B_k=A_k(1)$. In the round 1 (steps 3-4 of Algorithm~\ref{algo:LE}), the robots of $A_k(1)$ would move to $succ(v)$ and update $proceed=1$. In the round 2 (steps 6-7 of Algorithm~\ref{algo:LE}), the robots of $A_k(0)$ would move to $succ(v)$ and update $proceed=2$. In the round 3 (steps 9-12 of Algorithm~\ref{algo:LE}), the robots of $A_k(0)\cup A_k(1)$ would move back to the node $v$. In round 4 (steps 13-16 of Algorithm~\ref{algo:LE}), the robots of $A_k(0)$ would move to $pred(v_{t-1})$, and the robots of $A_k(1)$ would stay at the node $v_{t-1}$. In the round 5 (steps 18-21 of Algorithm~\ref{algo:LE}), the robots of $A_k(0)$ would return to $v_{t-1}$ and update $proceed=0$. The robots will stay at the current node for the rest of the rounds. 
\end{case}
The above three cases show that at the beginning of round 5 (steps 18-19 of Algorithm~\ref{algo:LE}) during the execution of $LeaderElection(M,k)$, the node $v_{t-1}$ will contain at least one robot from $G_{t-1}$. Thus, any robot of $G_t$ present at $v_{t-1}$ will not be able to identify $alone=true$ in round 5.

 \end{proof}

The following lemma proves that every chain merge to a single group in the chain.

\begin{lemma}\label{merge-complete}
Let $C_j=\lbrace G_1,G_2,\ldots,G_p\rbrace$ for some $j,p>0$ be a chain of groups in $\mathscr{R}$ where $G_i$ is a group of robots at node $v_i$. After the execution of $ActiveMerge$, the chain would be merged into a single group of robots (say $G$) within $O(p)$ synchronous rounds at the node $v_p$.
\end{lemma}
\begin{proof}
Initially, all the robots would update $status=ActiveMerge$. By the definition of a chain of groups, at least one empty node must exist between two chains of groups. From Lemma~\ref{leader_election_a} and \ref{leader_election_b}, it follows that within $O(MaxSize)$ rounds, a unique leader (say $L$) would be elected for the chain $C_j$. Note that the elected leader $L$ would remain at the node $v_1$ till $MaxSize$ phases are completed. Next, the leader would move forward in round 6 (steps 2-3). If $L$ identifies itself alone in round 7 (steps 5-6), it returns and terminates merging by updating $status=ActiveDisperse$. In round 8 (steps 8-9), if the non-leader robots identify $increase=false$, they follow the leader by moving forward. If the non-leader robots identify $increase=true$ in round 8 (steps 10-11), they identify that the leader $L$ has returned and terminated the merging. They also terminate the merging and update $status=ActiveDisperse$. As there is a total $k$ number of robots, at most, $k$ rounds are required to complete the merging of robots. As a consequence, all the robots in the chain of groups $C_j$ would merge into $G=G_1\cup G_2\cup\ldots\cup G_p$ within $O(MaxSize+k)$ synchronous rounds.

\end{proof}

The following lemma shows that the robots from two different chain remains separated and merged to different groups before they starts executing $ActiveDisperse$.

\begin{lemma}\label{merge-complete_b}
Let $C_i=\lbrace H_1,H_2,\ldots,H_q\rbrace$ and $C_j=\lbrace G_1,G_2,\ldots,G_p\rbrace$ be two chain of groups such that $u_q=pred(w_1), w_1=pred(w_2),\ldots w_l=pred(v_1)$ where $n-2\geq l\geq 1$. During the execution of subroutines $LeaderElection$ and $ActiveMerge$ no two robots from different chains would be co-located.
\end{lemma}
\begin{proof}

From Lemma~\ref{merge-complete}, it is guaranteed that within $O(MaxSize+k)$ synchronous rounds, all the robots in $C_i$ and $C_j$ would be merged, respectively. First, consider that $l=1$. We only need to show that during the execution of $LeaderElection$ and $ActiveMerge$, robots from different chains would remain separated. This is ensured for the following reasons:
\begin{enumerate}
\item Execution of $LeaderElection$: The robots (from the chain $C_i$) at $u_q$ would move to node $w_1$ in rounds 1-2 (steps 1-8 of Algorithm~\ref{algo:LE}). On the other hand, the robots (from the chain $C_j$) at $v_1$ would move to the node $w_1$ in rounds 4-5 (steps 13-21 of Algorithm~\ref{algo:LE}). Thus, robots from the chain $C_i$ and $C_j$ robots would not be co-located.
\item Execution of $ActiveMerge$: The leader (from the chain $C_i$) would move to the node $w_1$ (round 6, steps 2-4 of Algorithm~\ref{algo:AM}) and terminate the merging of groups at the node $u_q$ (round 7, steps 5-6 of Algorithm~\ref{algo:AM}). All the robots from the chain $C_j$ would not move to the node $w_1$ in rounds 6-8 (steps 2-12 of Algorithm~\ref{algo:AM}) during the execution of $ActiveMerge$. The leader and non-leader robots of the chain $C_j$ would similarly merge at the node $v_p$. Thus, robots from the chain $C_i$ and $C_j$ would not be co-located. 
\end{enumerate}
When $l\geq 2$, the robots from chain $C_i$ would move to $w_1$, and the robots from chain $C_j$ would move to $w_l$. As a consequence, robots from different chains would not be co-located.
\end{proof}
\begin{observation}
If a robot $M$ updates $status=ActiveDisperse$ from $status=ActiveMerge$ at the end of phase $t$, then for all the successive phases, $M$ would never update its status to $ActiveMerge$.
\end{observation}
\begin{observation}\label{occupy_one}
If a group of robots $G_i$ located at node $v_i$ updates $status=ActiveDisperse$ in phase $t$, the robots will occupy at most one vacant successive node in phase $t'>t$ when a split occurs.
\end{observation}

According to our proposed algorithm $ActiveMerge$, for two different chains, the merging of the groups in the respective chains will start in the same phase. However, since the length of the chains may be different, the merging may be completed in different rounds. Hence, one may think that when a group (after merging) starts $ActiveDisperse$ and expands by splitting into smaller groups, this expansion may meet the next group that is still merging.  The next lemma shows that this event can not happen. To be specific, when a group of robots, while executing $ActiveDisperse$ meet the next group, by that time, this next group must have started executing $ActiveDisperse$.

\begin{lemma}\label{merging-chains}
Let $C_i=\lbrace H_1,H_2,\ldots,H_q\rbrace$ and $C_j=\lbrace G_1,G_2,\ldots,G_p\rbrace$ be two chain of groups such that $u_q=pred(w_1), w_1=pred(w_2),\ldots, w_l=pred(v_1)$ where $n-2\geq l\geq 1$. If two robots, one from $C_i$, and one from $C_j$ ever become adjacent in some round $i$, then none of them can have $status=activemerge$ in round $i$. 
\end{lemma}
\begin{proof}
During the execution of $ActiveMerge$, all the robots would merge groups after $MaxSize$ rounds. If a chain of groups selects a leader before $MaxSize$ rounds, it would only start merging groups after $MaxSize$ rounds. Consider the case when $p\leq q$. As the chain $C_i$ has fewer groups of robots, the robots in $C_j$ would be guaranteed to merge before the robots would complete in $C_i$. Thus, a group of robots from $C_j$ would be merged to $C_i$ by ensuring that none of the robots have $status=activemerge$. Assume that $q<p$. Let $t$ be the earliest possible phase when the merging of $C_i$ is complete. At phase $t$, the robots of $C_i$ would update $status=activedisperse$. They would start the execution of $ActiveDisperse$ in phase $t+1$. As $q<p$, the robots in $C_j$ still require $p-q$ phases to merge. In each subsequent phase, the groups in $C_j$ would create one empty node until they complete the merging. From Observation~\ref{occupy_one}, the robots of $C_i$ would require at least $p-q$ rounds before it could merge a group of robots from $C_j$. Therefore, a group of robots from $C_j$ would be merged to $C_i$ by ensuring that none of the robots have $status=activemerge$. 

\end{proof}

\begin{remark}\label{lemma:not-move-back}
Let $M\in G_i$ be a robot of some chain $C_j$ that is present at a node $v_i$ at the start of any phase $t$. Consider the case when $M$ with $status=activedisperse$ or $status=passive$ identifies $increase=true$ due to incoming of a leader from a previous chain (say $C_i$). In round 12 (steps 10-11 of Algorithm~\ref{activedisperse} and steps 10-11 of Algorithm~\ref{passive}), $M$ would move to $pred(v_i)$ and update $status=passive$ in phase $t$. By moving to $pred(v_i)$ it becomes a part of the chain of groups $C_i$. For the incoming of a leader for each of the previous chain of groups, a robot $M$ would update $status=passive$. Thus, there can be at most $k$ number of such updates of $status=passive$. Otherwise, $M$ would either stay at $v_i$ or move to $succ(v_i)$ as discussed below:
\begin{enumerate}
    \item $M$ with $status=activedisperse$ would execute $ActiveDisperse$. During any execution of$ActiveDisperse$, there are two possible cases: a split happens, or no split happens. In both cases, the robots would only move to $succ(v_i)$. 
    \item $M$ with $status=passive$ would remain at $v_i$ by the execution of $Passive$. 
    \item $M$ with $status=jump$ would execute $Jump$. The robots would move to  $succ(v_i)$.
    \item $M$ with $status=waiting$ would remain at $v_i$ and wait for other robots to move forward.
\end{enumerate}
\end{remark}

The correctness of our algorithm depends on the fact that two consecutive groups of robots will not simultaneously splits into  smaller groups. This fact is proved in the next lemma. 

\begin{lemma}\label{lem:robot-possibility} 
During the execution of the algorithm $ActiveDisperse$, at the beginning of some phase $t$, let $w_1$ and $w_2$ be two consecutive occupied nodes. Let $M_1$ be a robot at $w_1$ and $M_2$ be a robot at $w_2$. Then exactly one of the following statement is true.
\begin{enumerate}[a.]
\item The status of $M_1$ is in \{activedisperse,wait,jump\} and the status of $M_2$ is passive.
\item The status of $M_2$ is in \{activedisperse,wait,jump\} and the status of $M_1$ is passive.
\end{enumerate}

\end{lemma}

\begin{proof}

{\it We prove the statements using mathematical induction.} \\
\noindent {\bf Base case:} Let $v_1$ be a node containing a group of robots at beginning of $ActiveDisperse$. Since by the definition of chain of groups, both the nodes $succ(v_1)$ and $pred(v_1)$ are empty after the chain merged to a single group at $v_1$, the statements are vacuously true. Consider the case when the leader present at $v_1$ founds a group from the successive chain of groups by moving two steps forward. The group of robots would move backward to the node $succ(v_1)$ and become $passive$. As a result, the statement $a$ is true at the beginning.

\noindent {\bf Inductive step:} Assume that one of the statements is true at the end of any phase $m>t$. We will prove that the statements must be true at the end of $(m+1)^{th}$ phase. From Remark \ref{lemma:not-move-back}, it follows that if any robot is present at a node $v_i$ at the end of phase $t$, then at the end of phase $t+1$ the robot can either stay at $v_i$ or move to $v_{i+1}$. Thus, we consider the cases of  robots with different status present at three consecutive nodes (say $v_i, v_{i+1}, v_{i+2}$) at the end of phase $m$. We have the following cases:
\begin{figure}[ht] 
\begin{center}
\subfloat[]
{\includegraphics[width=.28\textwidth]{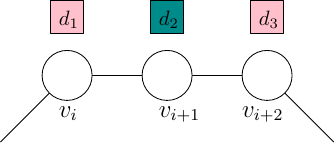}
\label{fig:ring-configa-1}}\hspace*{.5cm}
\subfloat[ ]
{\includegraphics[width=.28\textwidth]{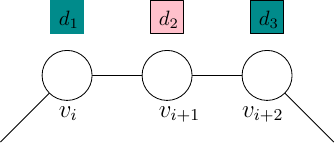}
\label{fig:ring-configa-2}\hspace*{.5cm}
}
\subfloat[ ]
{\includegraphics[width=.28\textwidth]{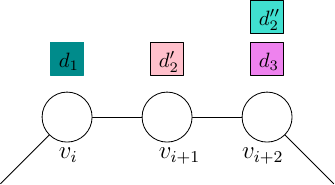}
\label{fig:ring-configa-3}
}\vspace*{.5cm}
\newline
\subfloat[ ]
{\includegraphics[width=.28\textwidth]{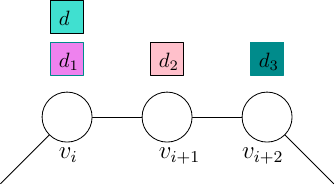}
\label{fig:ring-configa-4}\hspace*{.5cm}
}
\subfloat[ ]
{\includegraphics[width=.28\textwidth]{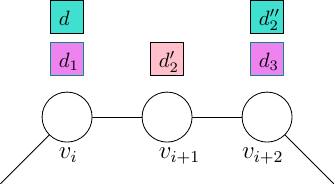}
\label{fig:ring-configa-5}
}
\end{center}
\caption{ Representation of the status of robots positioned on three consecutive nodes $v_i$, $v_{i+1}$, and $v_{i+2}$ (a) $passive$ (pink), $activedisperse$ (green) and $passive$ (pink), (b) $activedisperse$, $passive$, and $activedisperse$, (c) $activedisperse$, $passive$ and $wait$ (turquoise) and $jump$ (violet), (d) $wait$ and $jump$, $passive$ and $activedisperse$, (e) $wait$ and $jump$, $passive$ and $wait$ and $jump$. }
\label{fig:ring-configa}
\end{figure}
    
\setcounter{case}{0}

\begin{case}\normalfont
\textbf{The nodes $\boldsymbol{v_i}$, $\boldsymbol{v_{i+1}}$, and $\boldsymbol{v_{i+2}}$ contains $\boldsymbol{passive}$, $\boldsymbol{activedisperse}$, and $\boldsymbol{passive}$ robots, respectively.} Let there be $d_1$, $d_2$, and $d_3$ robots at the nodes $v_i$, $v_{i+1}$, and $v_{i+2}$, respectively, at the end of phase $m$ (see Figure \ref{fig:ring-configa-1}). At the end of phase $m+1$, the following sub-cases are to be considered:
\begin{enumerate}
\item {\bf No robot arrived at $v_i$ at the end of phase $m$.} Thus, all $d_1$ robots at $v_i$ become $activedisperse$ and remain at $v_i$. There are two possible scenarios at the node $v_{i+1}$:
\begin{enumerate}
\item No split happens at $v_{i+1}$. All the $d_2$ robots at $v_{i+1}$ remain at $v_{i+1}$ and their statuses become $passive$. Since no robots reach at $v_{i+2}$, all the $d_3$ robots at $v_{i+2}$ remain at $v_{i+2}$ and their statuses become $activedisperse$ (Figure \ref{fig:ring-configa-2}).
        
\item A split happens at $v_{i+1}$. A subset of $d_2$ robots (say $d'_2$) remain at $v_{i+1}$ and their statuses become $passive$. Since $d''_2 = d-d'_2$  robots reach at $v_{i+2}$, the $d''_2$ robots reach at $v_{i+2}$ and their statuses become $waiting$. Further, all the $d_3$ robots at $v_{i+2}$ remain at $v_{i+2}$ and their statuses become $jump$ (Figure \ref{fig:ring-configa-3}).
\end{enumerate}

\item {\bf Some robots reach at $v_i$ at the end of phase $m$.} Suppose $d$ robots reach at $v_i$ at the end of phase $m$. They update their status to $waiting$. All the $d_1$ robots at $v_i$ remain at $v_i$ and update their status to $jump$. There are two possible scenarios at the node $v_{i+1}$:
\begin{enumerate}
\item No split happens at $v_{i+1}$. All the $d_2$ robots at $v_{i+1}$ remain at $v_{i+1}$ and their statuses become $passive$. Since no robots reach at $v_{i+2}$, all the $d_3$ robots at $v_{i+2}$ remain at $v_{i+2}$ and their statuses become $activedisperse$ (Figure \ref{fig:ring-configa-4}).   
\item A split happens at $v_{i+1}$. Therefore, a subset of $d_2$ robots (say $d'_2$) remain at $v_{i+1}$ and their statuses become $passive$. As $d''_2 = d-d'_2$ robots reach at $v_{i+2}$, the $d''_2$ robots reach at $v_{i+2}$ and their statuses become $waiting$. Further, all the $d_3$ robots at $v_{i+2}$ remain at $v_{i+2}$ and their statuses become $jump$ (Figure \ref{fig:ring-configa-5}).  
\end{enumerate}
\end{enumerate}
\end{case}
\begin{figure}[ht] 
\begin{center}
\subfloat[ ]{\includegraphics[width=.3\textwidth]{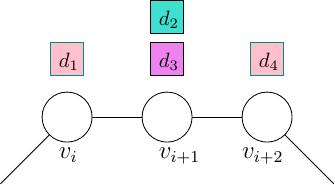}
\label{fig:ring-configb-1}
}
\subfloat[ ]{\includegraphics[width=.3\textwidth]{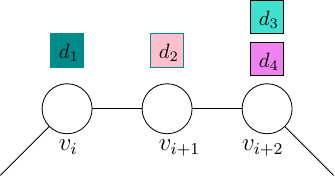}
\label{fig:ring-configb-2}
}
\subfloat[ ]{\includegraphics[width=.3\textwidth]{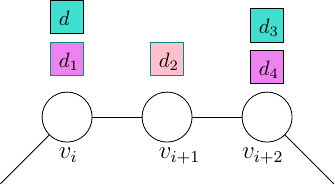}
\label{fig:ring-configb-3}
}
\end{center}
\caption{Representation of the status of robots positioned on three consecutive nodes $v_i$, $v_{i+1}$, and $v_{i+2}$ (a) $passive$, $wait$ and $jump$ and $passive$, (b) $activedisperse$, $passive$, and $wait$ and $jump$, (c) $wait$ and $jump$, $passive$, and $wait$ and $jump$. }
\label{fig:ring-configb}
\end{figure}
\begin{case}\normalfont
\textbf{The nodes \boldmath{$v_i, v_{i+1},\;and\; v_{i+2}$ contain $passive$, $\{jump,wait\}$, and $passive$ robots, respectively.}} Let the nodes $v_i, v_{i+1},v_{i+2}$ contain $d_1\;passive$ robots, $v_{i+1}$ contain $d_2\;waiting\;robots\;and\;d_3\; jump$ robots,  and $v_{i+2}$ contains $d_4$ $passive$ robots at the end of phase $m$ (Figure \ref{fig:ring-configb-1}). At the end of phase $m+1$, we have the following cases:
\begin{enumerate}
\item {\bf No robot arrived at $\boldsymbol{v_i}$ at the end of phase $\boldsymbol{m}$.} The $d_1$ robots at $v_i$ would update $status=activedisperse$. At the node $v_{i+1}$, the $d_2$ robots would stay and update $status=passive$, and the $d_3$ robots would move to the node $v_{i+2}$ and update status to $wait$. At the node $v_{i+2}$, the $d_4$ robots would update $status=jump$.
\item {\bf Some robots arrived at $\boldsymbol{v_i}$ at the end of phase $\boldsymbol{m}$.} Suppose $d$ robots arrive at $v_i$ at the end of phase $m$. The $d$ robots become $wait$ and the $d_1$ robots update $status=jump$. At the node $v_{i+1}$, the $d_2$ robots would stay and update $status=passive$, and the $d_3$ robots would move to the node $v_{i+2}$ and update status to $wait$. At the node $v_{i+2}$, the $d_4$ robots would update $status=jump$.
\end{enumerate}
\end{case}
\begin{figure}[ht] 
\begin{center}
\subfloat[ ]{\includegraphics[width=.3\textwidth]{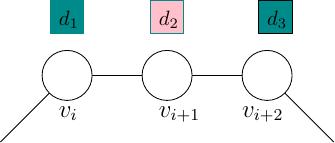}
\label{fig:ring-configc-1}
}
\subfloat[ ]{\includegraphics[width=.3\textwidth]{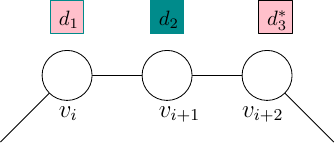}
\label{fig:ring-configc-2}
}
\subfloat[ ]{\includegraphics[width=.3\textwidth]{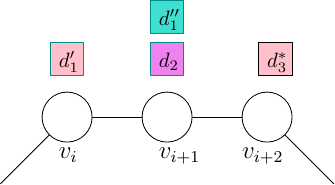}
\label{fig:ring-configc-3}
}
\end{center}
\caption{Representation of the status of robots positioned on three consecutive nodes $v_i$, $v_{i+1}$, and $v_{i+2}$ (a) $activedisperse$, $passive$, and $activedisperse$, (b) $passive$, $activedisperse$ and $passive$, (c) $passive$, $wait$ and $jump$, and $passive$.}
\label{fig:ring-configc}
\end{figure}
\begin{case}\normalfont
\textbf{The nodes \boldmath{$v_i, v_{i+1},\;and \; v_{i+2}$ contain $activedisperse$, $passive$, and $activedisperse$ nodes, respectively.}} Let there be $d_1$, $d_2$, and $d_3$ robots at the nodes $v_i$, $v_{i+1}$, and $v_{i+2}$, respectively, at the end of phase $m$. (Figure~\ref{fig:ring-configc}). Since the robots at $v_i$ are $activedisperse$, by the induction step, there would not be any incoming robots at $v_i$ (robots at $pred(v_i)$ must be $passive$). If there is no split at $v_i$, all the robots at $v_i$ become $passive$ at the beginning of $(m+1)^{th}$ phase. Consider the case when there is a split at $v_i$. A subset of $d_1$ robots (say $d'_1$) remain at $v_i$, and their statuses become $passive$. Since $d''_1 = d-d'_1$ robots reach at $v_{i+1}$, the $d''_1$ robots reach at $v_{i+1}$ and their statuses become $waiting$. The $d_2$ robots at $v_{i+1}$ become $jump$ for the $(m+1)^{th}$ phase. By the inductive step, $succ(v_{i+2})$ must be $passive$. The cases for the nodes $v_{i+2}$ and $succ(v_{i+2})$ would be similar to the cases for the nodes $v_i$ and $v_{i+1}$.
\end{case}
\begin{figure}[ht] 
\begin{center}
\subfloat[]{\includegraphics[width=.3\textwidth]{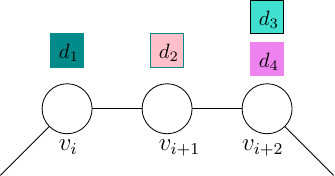}
\label{fig:ring-configd-1}
}
\subfloat[]{\includegraphics[width=.3\textwidth]{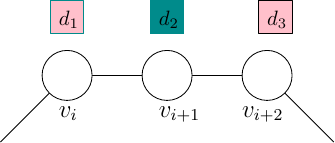}
\label{fig:ring-configd-2}
}
\subfloat[]{\includegraphics[width=.3\textwidth]{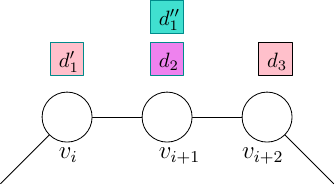}
\label{fig:ring-configd-3}
}
\end{center}
\caption{Representation of the status of robots positioned on three consecutive nodes $v_i$, $v_{i+1}$, and $v_{i+2}$ (a) $activedisperse$, $passive$, and $wait$ and $jump$, (b) $passive$, $activedisperse$ and $passive$, (c) $passive$, $wait$ and $jump$, and $passive$.}
\label{fig:ring-configd}
\end{figure}
\begin{case}\normalfont
    \textbf{The nodes \boldmath{$v_i, v_{i+1}, v_{i+2}$ contain $activedisperse$, $passive$ and $\{jump, wait \}$ robots, respectively.}} Let the nodes $v_i, v_{i+1},v_{i+2}$ contain $d_1\;activedisperse$ robots, $v_{i+1}$ contain $d_2\; passive$ robots, and $v_{i+2}$ contains $d_3\; jump$ robots, and $d_4\;wait$ robots at the end of phase $m$ (Figure \ref{fig:ring-configd}). Since the robots at $v_i$ are $activedisperse$, by the induction step, there would not be any incoming robots at $v_i$ (robots at $pred(v_i)$ must be $passive$). If there is no split at $v_i$, all the robots at $v_i$ become $passive$ at the beginning of $(m+1)^{th}$ phase. In this case, all the robots at $v_{i+1}$ become $activedisperse$ at the beginning of $(m+1)^{th}$ phase. Consider the case when there is a split at $v_i$. A subset of $d_1$ robots (say $d'_1$) remain at $v_i$, and their statuses become $passive$. Since $d''_1 = d-d'_1$ robots reach at $v_{i+1}$, the $d''_1$ robots reach at $v_{i+1}$ and their statuses become $waiting$. The $d_2$ robots at $v_{i+1}$ become $jump$ for the $(m+1)^{th}$ phase. Since the robots at $v_{i+1}$ are $passive$, there would not be any incoming robots at $v_{i+2}$. The $d_1$ robots would become $passive$, and $d_2$ robots would move to the node $succ(v_{i+2})$ and update status to $wait$ or $activedisperse$ based on whether the node $succ(v_{i+2})$ is occupied or empty. 
\end{case}
\begin{case}\normalfont
\textbf{The nodes \boldmath{$v_i, v_{i+1}, v_{i+2}$ contain $\{jump, wait \}$, \passive, and $activedisperse$ robots, respectively.}} Let the nodes $v_i, v_{i+1},v_{i+2}$ contain $d_1\;waiting$ robots and $d_2\; jump$ robots, $v_{i+1}$ contain $d_3\; passive$ robots, and $v_{i+2}$ contains $d_4\;activedisperse$ robots at the end of phase $m$ (Figure \ref{fig:ring-confige}). Since the robots at $v_i$ are $jump$ and $waiting$, by the induction step, there would not be any incoming robots at $v_i$ (robots at $pred(v_i)$ must be $passive$). The $d_1$ robots would become $passive$, and $d_2$ robots would move to the node $v_{i+1}$ and update status to $waiting$. The $d_3$ robots at $v_{i+1}$ would update status to $jump$. By the inductive step, $succ(v_{i+2})$ must be $passive$ as $d_4$ robots are $activedisperse$ at $v_{i+2}$. Depending on whether there is a split or not there would be two scenarios at the beginning of $(m+1)^{th}$ phase: $passive$ robots at $v_{i+2}$ and $activedisperse$ robots at $succ(v_{i+2})$, or $passive$ robots at $v_{i+2}$ and $jump$ and $wait$ robots at $succ(v_{i+2})$.
\end{case}

\begin{figure}[ht] 
\begin{center}
\subfloat[ ]{\includegraphics[width=.3\textwidth]{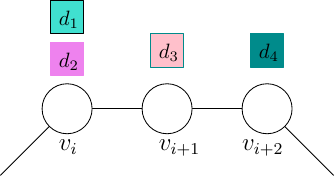}
\label{fig:ring-confige-1}
}
\subfloat[ ]{\includegraphics[width=.3\textwidth]{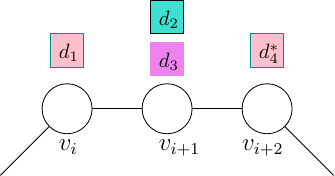}
\label{fig:ring-confige-2}
}
\end{center}
\caption{Representation of the status of robots positioned on three consecutive nodes $v_i$, $v_{i+1}$, and $v_{i+2}$ (a) $wait$ and $jump$, $passive$, and $activedisperse$, (b) $passive$, $wait$ and $jump$, and $passive$.}
\label{fig:ring-confige}
\end{figure}
 \begin{figure}[ht] 
\begin{center}
\subfloat[ ]{\includegraphics[width=.3\textwidth]{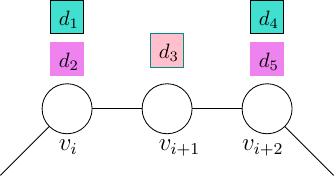}
\label{fig:ring-configf-1}
}
\subfloat[ ]{\includegraphics[width=.3\textwidth]{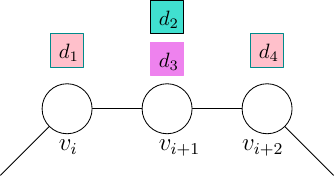}
\label{fig:ring-configf-2}
}
\end{center}
\caption{Representation of the status of robots positioned on three consecutive nodes $v_i$, $v_{i+1}$, and $v_{i+2}$ (a) $wait$ and $jump$, $passive$, and $wait$ and $jump$, (b) $passive$, $wait$ and $jump$, and $passive$.}
\label{fig:ring-configf}
\end{figure}

\begin{case}\normalfont
\textbf{The nodes \boldmath{$v_i, v_{i+1}, v_{i+2}$ contain $\lbrace jump,wait\rbrace$, \passive, and $\lbrace jump,wait\rbrace$ robots, respectively.}} Let the nodes $v_i, v_{i+1},v_{i+2}$ contain $d_1\;waiting$ robots and $d_2\; jump$ robots, $v_{i+1}$ contain $d_3\; passive$ robots, and $v_{i+2}$ contains $d_4\;waiting$ robots and $d_5\; jump$ robots at the end of phase $m$ (Figure \ref{fig:ring-configf}). Since the robots at $v_i$ are \jump~and wait, by the induction step, there would not be any incoming robots at $v_i$ (robots at $pred(v_i)$ must be $passive$). The $d_1$ robots would become $passive$, and $d_2$ robots would move to the node $v_{i+1}$ and update status to $waiting$. The $d_3$ robots at $v_{i+1}$ would update status to $jump$. By the inductive step, $succ(v_{i+2})$ must be $passive$ as $d_4$ robots are $jump$ and $waiting$ at $v_{i+2}$. The cases for the nodes $v_{i+2}$ and $succ(v_{i+2})$ would be similar to the cases for the nodes $v_i$ and $v_{i+1}$.
\end{case}
\begin{case}\normalfont
\textbf{$\boldsymbol{v_i}$ contains $\boldsymbol{activedisperse}$ statused robots with a leader, $\boldsymbol{v_{i+1}}$ is empty, and $v_{i+1}$ contains robots with status $\boldsymbol{activedisperse}$ or $\boldsymbol{passive}$ from a different chain of groups}. A leader present at $v_i$ would move forward only when the other robots at the node $v_i$ have $status=activedisperse$. Consider the case when the leader present at $v_i$ founds the group of robots at $v_{i+2}$ by moving two steps forward. The group of robots would move backward to the node $v_{i+1}$ and become $passive$. The statement $a$ is true at the beginning of phase $m+1$ because we have two adjacent nodes $v_i$ and $v_{i+1}$ where the status of the robots are $activedisperse$ and $passive$, respectively.   
\end{case}
\end{proof}

Two robots are called {\it distinguished} in a phase, if they are either not co-located at the beginning of a phase or if co-located, then their status is not the same.

The following lemma says that once two robots becomes distinguished once, them remains distinguished in the subsequent phases.

\begin{lemma}\label{distinguish_one}
If two robots are distinguished in some phase $t$, they remain distinguished in each phase $t'>t$.
\end{lemma}
\begin{proof}
Let $M_1$ and $M_2$ be two robots co-located at the node $v_0$ with $j^{th}$ bit 0 and 1, respectively. Suppose $M_1$ and $M_2$ are separated by processing the $j^{th}$ bitin phase $t$ by the execution of $activedisperse$. In the phase $t+1$, $M_1$ is at $v_o$ and $M_2$ is at $succ(v_0)$. As, $M_1$ and $M_2$ are not co-located at the same node, they are distinguished at phase $t+1$. Next, consider the case when $M_1$ and $M_2$ become co-located for some phase $t'>t+1$ at a node (say $v_i$). As the non-leader robots move only forward by one step in each phase, the robot $M_1$ must have moved from the node $pred(v_i)$ to $v_i$ in the phase $t'-1$. Thus, the robots at $pred(v_i)$ have either $status=activedisperse$ or a subset containing $M_1$ have $status=jump$ in the phase $t'-1$. By Lemma~\ref{lem:robot-possibility}, the robots at $v_i$ have $status=passive$ in the phase $t'-1$. In the phase $t'$, all the group of robots which moved to $v_i$ would have $status=wait$ (including $M_1$). The robots which were at the node $v_i$ would have $status=jump$. Thus, $M_1$ and $M_2$ are distinguished in the phase $t'$. Hence, if two robots are distinguished in some phase $t$, they remain distinguished in each phase $t'>t$.
\end{proof}

The dispersion is eventually achieved by splitting larger groups into smaller groups of robot which is done by processing the bits of the labels of the robots during subroutine $ActiveDisperse$. Hence, we have to make sure that sufficient numbers of $Activedisperse$ must be called during the execution of our algorithm. The following lemma proves the fact that the robots with status other than $activedisperse$ eventually changes their status to $activedisperse$ and hence guarantees the splitting as mentioned above.

\begin{lemma}\label{finite_change_a}
If in some phase, a robot has status $\in \{passive,  jump ,wait\}$, then within a finite number of phases, it changes its status to $activedisperse$.
\end{lemma}
\begin{proof}
If possible, a robot $M$ enters into an infinite loop of transitions of statuses as represented below:
\begin{center}
    $passive \longrightarrow jump \longrightarrow wait \longrightarrow passive$
\end{center}
In particular, a robot have $status=passive$ at some phase $t'$. By the execution of subroutine $Passive(M)$, $M$ updates $status=jump$ in round 17 (step 23 of Algorithm~\ref{passive}). In phase $t'+1$, $M$ updates $status=wait$ in round 17 (step 7 of Algorithm~\ref{jump}) by the execution of $Jump(M)$. In phase $t'+2$, $M$ updates $status=passive$ in round 17 (step 2 of Algorithm~\ref{wait}) by the execution of subroutine $Wait(M)$. Subsequently, the robot $M$ enters into a loop of three consecutive phase transitions $passive-jump-wait$.

Suppose $t$ denotes the phase at which the robot $M$ updates $status=passive$ before entering into the loop. Note that $M$ may have $status=activedisperse$ or $status=wait$ at the phase $t-1$. In every three consecutive phases of $passive-jump-wait$ cycle, a robot moves forward at least one step by the execution of $Jump(M)$ in round 14 (step 2 of Algorithm~\ref{jump}). As there are $k<n$ number of robots, after $k$ number of consecutive $passive-jump-wait$ phases, $M$ must find an empty node by the execution of subroutine $Jump(M)$ and update $status=activediaperse$. A contradiction. Hence, if a robot has status $\in \{passive,  jump ,wait\}$ in some phase, then within a finite number of phases, it changes its status to $activedisperse$.
\end{proof}

\begin{lemma}\label{jumpwait_k_times}
A robot, once changes its status to $activedisperse$ for the first time, can have
  status $wait$ in at most $k$ phases and status $jump$ in at most $2k$ phases.
\end{lemma}
\begin{proof}
    Let $M$ be a robot on a node $v_i$ with $status=activedisperse$ in some phase $t$. We have the following cases:
    \begin{enumerate}
        \item At most $k$ updates of $status=wait$. $M$ would update $status=wait$ in phase $t$ by the execution of $ActiveDisperse(M)$ (steps 29-30 of Algorithm~\ref{activedisperse}) if a split happened at $v_i$ and $succ(v_i)$ is occupied. If possible, $M$ updates $status=wait$ for more than $k$ phases. As there are total $k<n$ number of robots, after at most $k$ phases $M$ is guaranteed to find an empty node at which $M$ does not require to wait for other robots to move forward (do not need to update status to $wait$). A contradiction. Thus, a robot, once changes its status to $activedisperse$ for the first time, can have
       $status=wait$ in at most $k$ phases. 
        \item At most $k$ updates of $status=jump$. $M$ would update $status=passive$ in phase $t$ by the execution of $ActiveDisperse(M)$ if a split happened at $v_i$ and $M$ remained at $v_i$. In phase $t+1$, $M$ would update $status=jump$ by the execution of subroutine $Passive(M)$ in round 17 (step 23 of Algorithm~\ref{passive}) due to the incoming of robots from $pred(v_i)$. If possible, $M$ updates $status=jump$ for more than $k$ phases. As there are total $k<n$ number of robots, after at most $k$ phases it is guaranteed that there would not be any incoming robots to the current node of $M$. Thus, $M$ does not require to move forward to create space for other robots (do not need to update status to $jump$). A contradiction. Thus, a robot, once changes its status to $activedisperse$ for the first time, can have
       $status=jump$ in at most $k$ phases. 
    \end{enumerate}
In addition, in the light of Remark \ref{lemma:not-move-back}, a robot with status passive can change its status to passive at most $k$ times. Also, a robot with status $passive$ may change  its status to $jump$
 in the next phase. Hence, the total number of phases with status $jump$ is at most $2k$.   
\end{proof}

Now onwards, our technical results will converge in proving the final result that proves that dispersion is achieved and it is achieved in $O(MaxSize+k)$ rounds.

\begin{lemma}\label{distinguish_two}
    If a robot $M_1\in G$ in some chain $C$ has updated $status=activedisperse$ for $MaxSize$ number of phases, then for any other robot $M_2 \in G$, $M_1$ and $M_2$ must be distinguished.  
\end{lemma}
\begin{proof}
    If possible, let $M_1$ and $M_2$ remain co-located, and both have the same status after $MaxSize$ number of phases. As all the robots have unique IDs, $\exists j,\; 1\leq j\leq MaxSize$ such that the robots $M_1$ and $M_2$ have distinct $j^{th}$ bit. Consequently, they separated by processing the $j^{th}$ bit at some phase, say $t$. In phase $t+1$, $M_1$ and $M_2$ would not be co-located. Thus, they would become distinguished in phase $t+1$. By Lemma~\ref{distinguish_one}, they would remain distinguished for all phases $t'>t$. A contradiction. Hence, if a robot $M_1\in G$ in some chain $C$ has updated $status=activedisperse$ for $MaxSize$ number of phases, then for any other robot $M_2 \in G$, $M_1$ and $M_2$ must be distinguished.
\end{proof}

\begin{lemma}\label{passive_finite}
A robot can have passive status for at most $MaxSize+2k$ phases.
\end{lemma}
\begin{proof}
    A robot would update $status=passive$ by executing $ActiveDisperse(M)$ or $Wait(M)$. We have the following cases:
    \begin{enumerate}
        \item Execution of subroutine $ActiveDisperse(M)$. A robot $M$ with $status=activedisperse$ would update $status=passive$ by the execution of subroutine $ActiveDisperse(M)$ for the following cases: (a) a split happened, and $M$ remained at the current node, or (b) no split happened. As at most $MaxSize$ times $M$ would execute the subroutine $ActiveDisperse$, there can be at most $MaxSize$ number of $passive$ updates by the execution of subroutine $ActiveDisperse(M)$ (steps 23-26 of Algorithm~\ref{activedisperse}).
        \item Execution of subroutine $Wait(M)$. By Lemma~\ref{jumpwait_k_times}, it follows that there can be at most $k$ number of $status=wait$ updates and subsequently at most $k$ many $status=passive$ updates by the execution of $Wait(M)$ in round 17 (step 2 of Algorithm~\ref{wait}). 
    \end{enumerate}
    In addition, in the light of Remark \ref{lemma:not-move-back}, a robot with status $passive$ or $activedisperse$ can change its status to passive at most $k$ times. Hence, a robot can have $passive$ status for at most $MaxSize+2k$ phases.
\end{proof}

The final result is stated in the following theorem that proves that dispersion is achieved.

\begin{theorem}
Within $O(MaxSize+k)$ rounds after the start of $AlgorithmMultiStart$, no two robots on the ring are co-located.
\end{theorem}
\begin{proof}
Initially, all the robots have $status=leaderelection$. From Lemma~\ref{leader_election_b}, it follows that within $MaxSize$ rounds, a unique leader would be selected for each chain of groups. Lemma~\ref{merge-complete} ensures that a chain of length $p$ would be merged within $O(p)$ rounds. Lemma~\ref{merge-complete_b} ensures that no two robots from different chains would be co-located during the execution of $ActiveMerge$ and $LeaderElection$. By Observation 6, it follows that if a robot updates $status=activedisperse$ from $status=activemerge$, then for all the successive phases, the robot would never update its status to $activemerge$. From Lemma~\ref{merging-chains}, it follows that if two robots from different chains become adjacent, none of them can have $status=activemerge$. From Lemma~\ref{lem:robot-possibility}, robots with $status=passive$ would be guaranteed to be present in alternate nodes. Lemma~\ref{distinguish_one} ensures that once two robots become distinguished, they will remain distinguished for all the successive nodes. Lemma~\ref{finite_change_a} guarantees that if in some phase, a robot has status $\in \{passive, jump, wait\}$, then within a finite number of phases, it changes its status to $activedisperse$. It is guaranteed by Lemma~\ref{distinguish_two} that if a robot $M_1\in G$ in some chain $C$ has updated $status=activedisperse$ for $MaxSize$ number of phases, then for any other robot $M_2 \in G$, $M_1$ and $M_2$ must be distinguished. By the Lemma~\ref{jumpwait_k_times}, it follows that if a robot once changes its status to $activedisperse$ for the first time, it can have $status=wait$ in at most $k$ phases and $status=jump$ in at most $2k$ phases. Lemma~\ref{passive_finite} shows that a robot can have $status =passive$ for at most $MaxSize+2k$ phases. Hence, within $O(MaxSize+k)$ rounds after the start of $AlgorithmMultiStart$, no two robots on the ring are co-located. Since $MaxSize=\lfloor\log L\rfloor+1$, within $O(\log L+k)$ rounds after the start of $AlgorithmMultiStart$, no two robots on the ring will be co-located.
\end{proof}

\subsection{Lower Bound}
In this section, we show a lower bound $\Omega(k+\log L)$ on time for the dispersion of $k$ robots starting from one node of a cycle. Since there are $k$ robots, at least one must reach a node at a distance at least $\lceil\frac{k}{2}\rceil$ from the starting node. Since a robot can travel at most one edge in a round, at least $\lceil\frac{k}{2}\rceil$ many rounds are needed before dispersion is achieved.
This proves the lower bound $\Omega(k)$ on time of dispersion. 
Therefore, showing a lower bound $\Omega(\log L)$ is enough. The following lemma shows a lower bound $\Omega(\log L)$ on time of dispersion.

\begin{lemma}
Any algorithm that solves dispersion on a ring must take $\Omega(\log L)$ rounds.
\label{lem4}\end{lemma}
\begin{proof}
We prove this lemma by contradiction.
Suppose that there exists an algorithm $\mathcal{A}$ that solves the problem of dispersion in at most $y$ where $y \le \frac{\log L}{2}$ rounds. At any time $t>0$, a robot has three possible choices according to ${\cal A}$: (i) move clockwise from the current node, (ii) move counter-clockwise from the current node, and (iii) stay at the current node. Define a vector $X_i=(x_i^1,x_i^2, \ldots, x_i^{y})$ for the robot $M_i$ such that,

\centerline{$
x_i^j=\begin{cases}
			-1, & \text{if the robot $M_i$ moves clockwise at round $j$}\\
            1, & \text{if the robot $M_i$ moves counter-clockwise at round $j$}\\
            0, & \text{if the robot $M_i$ does not move at round $j$}
		 \end{cases}
$}
 
As all the robots are initially placed at the same node of the ring, if for any two robots $M_i$ and $M_j$, $X_i=X_j$, then these robots stay co-located even after $y$ rounds. There are at most  $3^{\frac{\log L}{2}} < L$ different possible vectors of length $y$. Also, there are at most $L$ different possible labels the set of robots may receive initially. Therefore, by Pigeon hole principle, there exist two integers $L_1, L_2 \le L$ such that if the two robots have labels $L_1$ and $L_2$, the respective movement vectors are the same. This shows that even after round $y$, the robots with labels $L_1$ and $L_2$ remain co-located. This contradicts the fact that dispersion can be solved within $y$ rounds.

\end{proof}

From Lemma \ref{lem4}, along with the lower bound $\Omega(k)$, we summarize the desired result in the following theorem.

\begin{theorem}
Any algorithm that solves dispersion on a ring must take $\Omega(\log L+k)$ rounds.
\end{theorem}
\section{Conclusions}
We have presented a deterministic solution for the dispersion problem in an oriented ring, allowing multiple initial points. Our proposed algorithm solves the dispersion problem within $O(\log L+k)$ rounds. We have presented a lower bound $\Omega(\log L+k)$ on time for the dispersion of $k$ robots on a ring. The assumption of a common orientation has a significant impact on the results. It would be interesting to study the problem in an unoriented ring. Another assumption is about the knowledge of $L$, which has helped design our proposed algorithm. An algorithm without the assumption of knowledge of $L$ can be explored in the future. The dispersion problem with multiple sources can be investigated for other kinds of graphs, namely cactus, grids, torus, and general graphs.

\bibliographystyle{plainurl}
\bibliography{biblio}

\end{document}